\documentclass{article}
\usepackage{geometry}                
\usepackage{graphicx}
\usepackage{stmaryrd}
\usepackage{amssymb}
\usepackage{amsthm, bm}
\usepackage{amsmath, cancel, centernot }
\usepackage[mathscr]{eucal}
\usepackage{mathtools}
\usepackage{epstopdf}
\title{Superdeterminism Without Conspiracy}
\author{Tim Palmer \\ Department of Physics, University of Oxford, UK}
\date{\today}                                          
\makeatletter
\newcommand\be{\@ifstar{\[}{\begin{equation}}}
\newcommand\ee{\@ifstar{\]}{\end{equation}}}
\newcommand\bp{\begin{pmatrix}}
\newcommand\ep{\end{pmatrix}}

\newtheorem*{theorem*}{Theorem}

\newtheorem*{corollary}{Corollary}
\makeatother
\begin{document}
\bibliographystyle{plain}
\maketitle

\begin{abstract}

Superdeterminism - where the Measurement Independence assumption in Bell's Theorem is violated - is frequently assumed to imply implausibly conspiratorial correlations between properties $\lambda$ of particles being measured and measurement settings $x$ and $y$. But it doesn't have to be: a superdeterministic but non-conspiratorial locally causal model is developed where each pair of entangled particles has unique $\lambda$. The model is based on a specific but arbitrarily fine discretisation of complex Hilbert space, where $\lambda$ defines the information, over and above the freely chosen nominal settings $x$ and $y$, which fixes the exact measurement settings $X$ and $Y$ of a run of a Bell experiment. Pearlean interventions, needed to assess whether $x$ and $y$ are Bell-type free variables, are shown to be inconsistent with rational-number constraints on the discretised Hilbert states. These constraints limit the post-hoc freedom to vary $x$ keeping $\lambda$ and $y$ fixed but disappear with any coarse-graining of $\lambda$, $X$ and $Y$, rendering so-called drug-trial conspiracies irrelevant. Points in the discretised space can be realised as ensembles of symbolically labelled deterministic trajectories on an `all-at-once' fractal attractor. It is shown how quantum mechanics might be `gloriously explained and derived' as the singular continuum limit of the discretisation of Hilbert space;  It is argued that the real message behind Bell's Theorem has less to do with locality, realism or freedom to choose, and more to do with the need to develop more explicitly holistic theories when attempting to synthesise quantum and gravitational physics. 
\end{abstract}

\section{Introduction}
\label{intro}

A deterministic hidden-variable model is said to be superdeterministic - not a word the author would have chosen - if the so-called Measurement Independence assumption (sometimes referred to as the Statistical Independence assumption or the $\lambda$-independence assumption),
\be
\label{statind}
\rho(\lambda | xy)=\rho(\lambda)
\ee
is violated \cite{Hall:2010} \cite{Scarani:2019} \cite{Chen:2020} \cite{HossenfelderPalmer} \cite{tHooft:2015b}. Here $\rho$ is a probability density on a set of hidden variables $\lambda$, and $x$ and $y$ $\in \{0,1\}$ denote experimentally chosen measurement settings - for concreteness, nominally accurate polariser orientations. Without (\ref{statind}), it is impossible to show that model satisfies the CHSH version of Bell's inequality
\be
\label{be}
|C(x=0, y=0)-C(x=0, y=1)+C(x=1, y=0)+C(x=1, y=1)| \le 2
\ee
where $C$ denotes a correlation on Bell-experiment measurement outcomes over an ensemble of particle pairs prepared in the singlet state. 

The argument that models which violate (\ref{statind}) are conspiratorial, originates in a paper by Shimony, Horne and Clauser, written in response to Bell's paper on local beables \cite{Belletal:1985}. Shimony \emph{et al} write:
\begin{quote}
In any scientific experiment in which two or more variables are supposed to be randomly selected, one can always conjecture that some factor in the overlap of the backwards light cones has controlled the presumably random choices. But, we maintain, skepticism of this sort will essentially dismiss all results of scientific experimentation. Unless we proceed under the assumption that hidden conspiracies of this sort do not occur, we have abandoned in advance the whole enterprise of discovering the laws of nature by experimentation.
\end{quote}
The drug trial is often used to illustrate the contrived nature of such a conspiracy. For example \cite{Goldsteinetal:2011}:
\begin{quote}
$\dots$ if you are performing a drug versus placebo clinical trial, then you have to select some group of patients to get the drug and some group of patients to get the placebo. The conclusions drawn from the study will necessarily depend on the assumption that the method of selection is independent of whatever characteristics those patients might have that might influence how they react to the drug.
\end{quote}
Related to this, superdeterminism is sometimes described as requiring exquisitely (and hence unrealistically) finely tuned initial conditions \cite{Baas:2020}, or as negating experimenter freedom \cite{Zeilinger}. A number of quantum foundations experts (e.g. \cite{Maudlin:2023} \cite{WisemanCavalcanti} \cite{Araujo:2019} \cite{Aaronson}) use one or more of these arguments to dismiss superdeterminism in derisive terms.

A new twist was added by Aaronson \cite{Aaronson} who concluded his excoriating critique of superdeterminism with a challenge:
\begin{quote}
Maxwell's equations were a clue to special relativity. The Hamiltonian and Lagrangian formulations of classical mechanics were clues to quantum mechanics. When has a great theory in physics ever been grudgingly accommodated by its successor theory in a horrifyingly ad-hoc way, rather than gloriously explained and derived?
\end{quote} 

It would seem that developing superdetermistic models of quantum physics is a hopeless cause. However, the purpose of this paper is to attempt to show that superdeterminism has been badly misunderstood, to rebuff these criticisms and indeed suggest that violating (\ref{statind}) is perhaps the \emph{only} sensible way to understand the experimental violation of Bell inequalities. Importantly, we show that whilst conspiracy would imply a violation of (\ref{statind}), the converse is not true. In Section \ref{consp}, we define what we mean by a non-conspiratorial violation of (\ref{statind}) and how it differs from these more traditional conspiratorial violations. Motivated by this, and the unshieldable effects of gravity as described in Section \ref{andromeda}, a non-conspiratorial superdeterministic model is described in Section \ref{superdeterministic}, based on a specific discretisation of complex Hilbert Space \cite{Palmer:2020}. Using the homeomorphism between $p$-adic integers and fractal geometry, the model is linked to the invariant set postulate \cite{Palmer:2009a} - the universe is evolving precisely on some special dynamically invariant subset of state space. In Section \ref{belltheorem} it is shown how this model violates (\ref{statind}) non-conspiratorially, and indeed violates (\ref{be}) in exactly the same way as does quantum mechanics. In Section \ref{lc} we show that the superdeterministic model is locally causal. In Section \ref{other} we discuss common objections to superdeterminism including fine tuning, free will and the so-called drug-trial analogy. Addressing Aaronson's challenge in Section \ref{aaronson}, we show how the state space of quantum mechanics can be considered the singular continuum limit of the discretised Hilbert space of the superdeterministic model. A possible experimental test of the supderdeterministic model is discussed in Section \ref{experiment}. 

Before embarking on this venture, one may ask answer the question: why bother? After all, quantum mechanics is an extremely well tested theory, and has never been found wanting. Why not just accept that quantum theory violates the concept of local realism - whatever that means -  and get on with it? The author's principal motivation for pursuing theories of physics which violate (\ref{statind}) lies in the possibility of finding a theory of quantum physics that is not inconsistent with the locally causal nonlinear geometric determinism of general relativity theory. As discussed in Section \ref{conclusions}, results from this paper suggest that instead of seeking a quantum theory of gravity (`quantum gravity') we should be seeking a strongly holistic gravitational theory of the quantum, from which the Euclidean geometry of space-time is emergent as a coarse-grained approximation to the $p$-adic geometry of state space. This, the author believes, is the real message behind the experimental violation of Bell inequalities. 

\section{Conspiratorial and Non-conspiratorial Interventions}
\label{consp}

There is no doubt that conspiratorial violations of Bell inequalities, of the type mentioned in the Introduction, imply a violation of (\ref{statind}). Here we are concerned with the converse question: does a violation of (\ref{statind}) imply the existence of a conspiratorial hidden variable theory? In preparing to answer this question, we quote from Bell's response `Free Variables and Local Causality' \cite{Belletal:1985} (FVLC) to Shimony et al \cite{Belletal:1985}. In FVLC, Bell writes:

\begin{quote}
I would insist here on the distinction between analyzing various physical theories, on the one hand, and philosophising about the unique real world on the other hand. In this matter of causality it is a great inconvenience that the real world is given to us once only. We cannot know what would have happened if something had been different. We cannot repeat an experiment changing just one variable; the hands of the clock will have moved and the moons of Jupiter. Physical theories are more amenable in this respect. We can calculate the consequences of changing free elements in a theory, be they only initial conditions, and so can explore the causal structure of the theory. I insist that B [Bell's paper on the theory of local beables \cite{Belletal:1985}] is primarily an analysis of certain kinds of physical theory.
\end{quote}
To understand the significance of this quote, we base the analysis of this paper around the thought experiment devised by Bell in FVLC, where, by design, human free will plays no explicit role. Bell supposes $x$ and $y$ are determined by the outputs of two pseudo-random number generators (PRNGs). These outputs are sensitive to the parity $P_x$ and $P_y$ of the millionth digits of the PRNG inputs. Bell now makes what he calls a `reasonable' assumption:
\begin{quote}
But this peculiar piece of information [whether the parity of the millionth digit is odd or even] is unlikely to be the vital piece for any distinctively different purpose, i.e., it is otherwise rather useless $\ldots$ In this sense the output of such a [PRNG] device is indeed a sufficiently free variable for the purpose at hand.
\end{quote}
It is important to note that, in this quote, Bell deflects discussion away from statistical properties of some ensemble of runs of an experiment where measurement settings are supposedly selected randomly (as per the Shimony et al quote above), and focusses on one individual run of an experiment. There is an important reason for this. When discussing conspiratorial hidden-variable models of the Shimony et al type, it is assumed that in any large enough ensemble with common value of $\lambda$, there exist sub-ensembles for each of the four pairs of measurement settings ($00$, $01$, $10$ and $11$). In this context, (\ref{statind}) implies that the four sub-ensembles are statistically equal. Conversely, in a conspiratorial violation of (\ref{statind}), the four sub-ensembles are statistically unequal. In such a situation, $\rho(\lambda | xy)$ can be interpreted as a frequency of occurrence within each of the four sub-ensembles. It is worth noting that in such a frequency-based interpretation of $\rho$, the issue of counterfactual definiteness - a central issue below - never arises. This has led to a misconception that counterfactual definiteness plays no role in Bell's Theorem.  

Importantly, hidden-variable models do not have to be like this. It is possible that the value of $\lambda$ is unique to each run of a Bell experiment. The model described below has this property. In this situation, if $\rho(\lambda| xy)$ were to define a frequency of occurrence, and a particle with value $\lambda$ was measured with settings $x$ and $y$ and could only be measured once, then $\rho(\lambda |xy) \ne 0$ and $\rho(\lambda |x'y)=\rho(\lambda |xy') =\rho(\lambda |x'y') =0$ where $x'=1-x$ and $y'=1-y$. However, this does not itself imply a violation of (\ref{statind}) - it merely emphasises that $\rho(\lambda |xy)$ is fundamentally \emph{not} a frequency of occurrence in an ensemble but rather is a \emph{probability density} defined on an individual particle with value $\lambda$. 

With that in mind, let us continue to focus, as Bell does, on a single entangled particle pair.  If $x$ and $y$ were not free variables, then $P_x$ or $P_y$ would not be vital for 'distinctively different' purposes. That is to say, we could vary $P_x$ or $P_y$ without having a vital impact on distinctly different systems. In the language of Pearl's causal inference theory \cite{Pearl}, if $x$ and $y$ were not free variables, there would exist small so-called interventions which by design changed $P_x$ or $P_y$, and by consequence had a vital impact on distinctly different systems. 

The most important part of this paper is to draw attention to two possible ways this might happen. The first is the conventional way where the effect of the intervention propagates causally from its localised source in space-time, somehow vitally influencing distinctly different systems. It is hard to imagine how varying something as insignificant as the parity of the millionth digit of an input to a PRNG could have such a vital impact. For this reason, Bell argued, the PRNG output should indeed be considered a free variable. This, of course, is not unreasonable.  

However, there is a second possibility - one that was not considered by Bell - that such interventions are simply inconsistent with physical theory. That is to say, the hypothetical state of the universe where $P_x$ or $P_y$ is perturbed but all other distinctly different elements of the universe are kept fixed, is inconsistent with the laws of physics. If such an intervention was hypothetically applied to a localised region of the universe, the ontological status of the whole universe would change; clearly a state of the universe as a whole only exists if all parts of it satisfy the laws of physics. Of course, we cannot perform an actual experiment to test directly this potential inconsistency: in changing the millionth digit, the hands of the clock and the moons of Jupiter will have moved, as Bell notes. Hence addressing the question of whether the PRNG output is a free variable in the sense of this paragraph requires studying the mathematical properties of physical theory. This is Bell's point in the first quote and it is the topic of this paper. 

Below we develop a model where each particle pair has unique $\lambda$ with the property
\be
\label{statind2}
\rho(\lambda | xy) \ne 0 \implies  \rho(\lambda | x'y)=0, \ \rho(\lambda, xy')=0; \ \ \ x'=1-x \ \ y'=1-y
\ee
As will be shown this can be interpreted as a locally causal non-conspiratorial violation of (\ref{statind}). It implies that an intervention on $x$, keeping $\lambda$ and $y$ fixed, leads to a state of the world which is inconsistent with the model postulates and therefore has zero probability. Clearly, this cannot be an intervention within space-time. Instead the intervention describes a hypothetical perturbation which takes a point in the state space of the universe (labelled by the triple $(\lambda, x, y)$), consistent with the model and hence with $\rho(\lambda |xy) \ne 0$, to a state $(\lambda, x',  y)$ which is inconsistent with physical theory and hence has $\rho(\lambda | x'y) = 0$. Importantly, this means a non-conspiratorial interpretation of (\ref{statind2}) implies that physical theory does not have the post-hoc property of counterfactual definiteness. The essential nature of counterfactuals in Bell's Theorem was pointed out by Redhead \cite{Redhead} some years ago. This important point seems to have been lost in more recent discussions of Bell's Theorem. 

However, one needs to be careful to not throw the baby out with the bathwater. Counterfactual reasoning is both pervasive and important in physics. Indeed, it is central to the scientific method \cite{Maudlin}. The fact that we can express laws of physics mathematically gives us the power to estimate what would have happened had something been different. Such estimates can lead to predictions and the predictions can be verified by experiment. We clearly do not want to give up counterfactual reasoning entirely in our search for new theories of physics. We address this concern by noting that output from experiment and physical theory (particularly when the complexities of the real world are accounted for) involves some inherent coarse-graining. Experiments have some nominal accuracy and output from a computational model is typically truncated to a fixed number of significant digits. We can represent such coarse graining by integrating exact output of physical theory over small volumes $V_\epsilon>0$ in state space, where $V_\epsilon \rightarrow 0$ smoothly as $\epsilon \rightarrow 0$. In developing a superdeterministic model below, we will require that counterfactual definiteness holds generically when variables are coarse grained over volumes $V_\epsilon$, for suitably defined small values $\epsilon >0$. The model based on discretised Hilbert Space, described in Section \ref{superdeterministic} has this property. As discussed below, this renders the drug-trial analogy irrelevant. 

These matters are subtle, and it seems Bell appreciated this. Instead of derisively rejecting theories where (\ref{statind}) is violated he concludes FVLC with the words:
\begin{quote}
Of course it might be that these reasonable ideas about physical randomisers are just wrong - for the purposes at hand.
\end{quote}
Indeed, in his last paper `La Nouvelle Cuisine', Bell writes:
\begin{quote}
An essential element in the reasoning here is that [polariser settings] are free variables.... Perhaps such a theory could be both locally causal and in agreement with quantum mechanical predictions. However, I do not expect to see a serious theory of this kind. I would expect a serious theory to permit `deterministic chaos' or `pseudorandomness' \ldots But I do not have a theorem about that.
\end{quote}
The last sentence is insightful, because, as we discuss, there is no such theorem. Indeed the reverse: here we develop a serious non-classical model where polariser settings are not free variables, utilising geometric concepts in deterministic chaos. 

\section{The Andromedan Butterfly Effect}
\label{andromeda}

The purpose of this section is to note how the Principle of Equivalence makes the interaction of matter with gravity especially chaotic. Here we repeat a calculation first reported by Michael Berry \cite{Berry:1978} and further analysed in \cite{Schwartz:2019}. It is well known that the flap of a butterfly's wings in Brazil can cause a tornado in Texas. But, could the flap of a butterfly's wings on a planet in the Andromedan galaxy cause a tornado in Texas?

We consider molecules in the atmosphere as hard spheres of radius $R$ with mean free distance $l$ between collisions. We wish to estimate the uncertainty $\Delta \theta_M$ in the angle of the $M$th collision of one of the spheres with other spheres, due to some very small uncertain external force. It is easily shown that $\Delta \theta_M$ grows exponentially with $M$. In particular
\be
\Delta \theta_M \approx (\frac{l}{R})^M \Delta \theta_1
\ee

With $l \approx 10^{-7}m$ and $R \approx 10^{-10}$ m, $l/R\approx 10^3$. After how many collisions $M$ is the position of a molecule in Earth's atmosphere sensitive to the gravitational uncertainty in the uncertain flap of a butterfly's wing in the Andromeda galaxy? Let $r$ denote the distance between Earth and Andromeda. The flap of a butterfly's wing through a distance $\Delta r$ will change the gravitational force on our target molecule by an amount $\Delta F = \partial F/ \partial r \ \Delta r =(Gm_{mol} m_{wing}/r^3) \ \Delta r$. Uncertainty in the flap of an Andromedan butterfly's wing will therefore induce an uncertainty in the acceleration of a terrestrial atmospheric molecule by an amount 
\be
\Delta a= \frac{2Gm_{wing}}{r^3} \Delta r
\ee
If $\tau$ denotes the mean time between molecular collisions, there is an uncertainty in the direction of the molecule $\Delta \theta_1 \approx \tau^2 \Delta a / R$. Plugging in $m_{wing}=10^{-5}$ kg, $\Delta r= 10^{-2}$m, $\tau = 10^{-9}$s and $r=10^{22}$m, gives $\Delta \theta_1= 10^{-90}$. How large is $M$ before $\Delta \theta_M \approx 1$? From the above, $10^{3M} \approx 10^{90}$. Hence after about 30 collisions the direction the direction of travel of the terrestrial molecule has been rendered completely uncertain by the gravitational effect of the Andromedan butterfly. Indeed one can go further: the direction of the terrestrial molecule is rendered completely uncertain by the uncertain position of a single electron at the edge of the visible universe after about only 50 or so collisions (which below we will round up to an order of magnitude of $10^2$). Once the molecules of the Earth's atmosphere have been disturbed in this way, it is only a matter of a couple of weeks before the nonlinearity of the Navier-Stokes equations leads to uncertainty in a large-scale weather pattern, such as a Texan tornado \cite{Lorenz:1969} \cite{Palmer:2022b}.

Nowhere above did the mass of the molecule enter the calculation. The same calculation could just as well apply to the measurement process in quantum physics - where a particle interacts with atoms in some measuring device. Indeed the same calculation could apply to molecules in an experimenter's brain, affecting the decisions they make. 

As a matter of principle, the direction of motion of a molecule after 100 collisions is for all practical purposes uncomputable. If a computation (say on a supercomputer) was attempted, then the direction of the molecule would depend on the motion of electrons in the chips of the supercomputer. Self-referentially, the computational software would have to include a representation of the computation. Technically, the Andromedan Butterfly Effect describes a computationally irreducible system \cite{Wolfram} - one that cannot be computed with a computationally simpler system. 

This is the first step in our argument that (\ref{statind}) could be violated because the universe should be considered a rather holistic chaotic system.  Gravity is what makes the universe holistic. Unlike the other forces of nature, the effects of gravity cannot be shielded by negative charges. However, by itself, this argument doesn't imply that the parity of the millionth digit, or the flap of a butterfly's wings in Andromeda is a \emph{vital} piece of information for determining distinctly different systems. For that we need more from the theory of chaos. 

\section{A Superdeterministic Model}
\label{superdeterministic}

\subsection{Nominal vs Exact Measurement Settings and Hidden Variables}

We build on Bell's thought experiment whereby polariser orientations are determined by the parities $P_x$ and $P_y$ in Alice and Bob's PRNGs (supposing there is no particular reason why these would be odd or even). Manifestly, $P_x$ and $P_y$ only determine the polariser orientations to some nominal accuracy. That is, $x, y \in \{0,1\}$ determine small neighbourhoods of the 2-sphere of orientations, referred to as $\epsilon$ disks. None of the results below depend on the magnitude of $\epsilon$ as long as $\epsilon > 0$ (as discussed in Section \ref{aaronson}, the limit $\epsilon=0$ in the proposed model is singular). It will be assumed - consistent with our search for an underlying deterministic theory - that when measurements are made on particular particles, the measurement outcomes $\pm 1$ are associated with some exact, albeit unknown, polariser orientations $\bold X$ and $\bold Y$. Here, $\bold X$ and $\bold Y$ denote unit vectors in physical 3-space from the centre of a unit ball, to the surface of the 2-sphere of orientations. The corresponding points on the 2-sphere to which the vectors point are written as $X$ and $Y$. The nominal directions $\bold x$ and $\bold y$ refer to unit vectors pointing to the centroids of the $\epsilon$ disks. In the discussion below, we assume that the probabilistic nature of the quantum measurement problem arises because the measurement outcome is typically sensitive to the exact measurement settings within an $\epsilon$ disk (c.f. fractally intertwined basins of attraction \cite{Palmer:1995}). Hence, any given nominal setting is consistent with an ensemble of possible exact settings. 

With this in mind, we let $\lambda$ describe all of the variables which, over and above $P_x$ and $P_y$, determine the exact measurement settings $\bold X$ and $\bold Y$. These variables include Andromedan butterflies and electrons at the end of the visible universe. When measuring a single qubit we can frame it like this: consider a spacelike hypersurface $\mathcal S$ in the past of some event $\mathscr E$ where the measurement outcome was known, and before the event where the PRNG output was known, then $\lambda$ must include data on $\mathcal S$, on and inside $\mathscr E$'s past light cone, with the exception of $P_x$ (or its determinants). In this way we can write $\bold X=\bold E(\lambda, \bold x)$, $\bold Y= \bold E(\lambda, \bold y)$. The extension of this causal picture for entangled qubits is discussed in Section \ref{lc}. 

\subsection{Rational Quantum Mechanics: RaQM}
\label{raqm}

Motivated by John Wheeler's plea to excise the continuum from physical theory \cite{Wheeler} - see also \cite{Ellis:2018} - a way to introduce non-conspiratorial superdeterminism into quantum physics is to discretise Hilbert space \cite{Buniy:2005} \cite{Buniy:2006} \cite{Palmer:2020} \cite{Carroll:2023}. At the experimental level this is surely unexceptionable: all experiments which confirm quantum mechanics will necessarily confirm a model of quantum physics based on discretised Hilbert Space, providing the discretisation is fine enough. However, as discussed, such a discretisation has profound implications for the interpretation of quantum experiments

\subsubsection{Single Qubits}
\label{sq}

Consider a qubit prepared in a state  $|1\rangle$, and written as 
\be
\label{qubit2}
|1\rangle = \cos\frac{\theta}{2} |1'\rangle + \sin\frac{\theta}{2} e^{i \phi} | -1' \rangle
\ee
with respect to some arbitrary basis $\{|1'\rangle, |-1'\rangle\}$. 
We call this a `proper basis' if  
\be
\label{semiexact}
\cos^2 \frac{\theta}{2} =\frac{m_1}{p};\ \ \ \ \frac{\phi}{2\pi} =\frac{n_1}{p}
\ee
where $p$ is some large prime number and $0\le m_1, n_1 \le p$. In a proper basis, it can be shown \cite{Palmer:2022} that $|\psi(\theta, \phi)\rangle$ has a representation as a bit string of $p$ deterministic elements $\pm 1$ where $\cos^2 \theta/2$ equals the fraction of elements $+1$ in the bit string and $\phi$ denotes a cyclical permutation of elements of the bit string. Here we will assume $p$ is a large prime. In RaQM the measurement basis is a proper basis, and the measurement output is a deterministically selected element of the bit string. In this way, each proper basis is a basis with respect to which a measurement could potentially be made by the experimenter. A measurement cannot be made with respect to a basis which is not proper. For each proper basis, the potential measurement outcome is associated with some exact measurement orientation $(\theta_e, \phi_e)$ in physical 3-space, which satisfies the rationality constraints
\be
\label{exact}
\cos^2 \frac{\theta_e}{2} =0.m_1m_2 \ldots m_I \in \mathbb Q \implies \cos \theta_e \in \mathbb Q;\ \ \  \ \ \frac{\phi_e}{2\pi} =0.n_1n_2\ldots n_J \in \mathbb Q
\ee
written in base-$p$. In this way, $(\theta, \phi)$ correspond to nominally accurate measurement orientations. In this representation, measurement outcome probabilities (over ensembles of elements of the bit string) automatically satisfy Born's Rule. Hence Born's Rule is not a separate axiom in RaQM. 

If the exact preparation setting is represented by $\bold X$, and exact measurement setting $\bold Y$, then the first of (\ref{exact}) can be written $\bold X \cdot \bold Y \in \mathbb Q$, where $\cdot$ denotes the scalar product of two vectors. 
Because of Niven's theorem
\begin{theorem*}
\label{niven}
 If $\phi_e/2\pi \in \mathbb Q$, then $\cos \phi_e \notin \mathbb{Q}$ except when $\cos \phi_e=0, \pm \frac{1}{2}, \pm 1$ \cite{Niven} \cite{Jahnel:2005}
\end{theorem*}
\noindent angles that determine probabilities and angles that determine phases, c.f., (\ref{exact}), are relatively incommensurate, except at the \emph{precise} values $\phi_e =0, \pi/2, \pi, 3\pi/2$. One can assume that such precise values never occur in practice (any gravitational wave would disturb a system away from such a precise value), though they may be relevant for theoretical reasons (e.g. when one considers a measurement performed with a \emph{precisely} opposite measurement direction). 

An important corollary to Niven's Theorem - central to this paper - is what we refer to as the Impossible Triangle Corollary:
\begin{corollary}
Let $\triangle{XYZ}$ be a triangle on the unit sphere with rational internal angles, not precisely equal to multiples of $45^\circ$, such that $\bold X \cdot \bold Y$ and $\bold Y \cdot \bold Z \in \mathbb Q$. Then $\bold X \cdot \bold Z \notin \mathbb Q$.  
\end{corollary}
\begin{proof}
Assume $\bold X \cdot \bold Z= \cos \theta_{XZ}$ is rational, where $\theta_{XZ}$ denotes the exact angular distance between $X$ and $Z$ on the unit sphere, etc. By the cosine rule for $\triangle {XYZ}$,
\be
\label{cosinerule}
\cos \theta_{XZ}= \cos \theta_{XY} \cos \theta_{YZ} + \sin \theta_{XY} \sin \theta_{YZ} \cos\phi_{Y}
\ee
where $\phi_{Y}$ is the exact internal angle of the triangle at the vertex $Y$. Since $\cos\theta_{XY}$ and $\cos \theta_{YZ}$ are both rational, then from (\ref{cosinerule}), $\sin \theta_{XY} \sin \theta_{YZ} \cos \phi_{Y}$ must be rational. Squaring, $(1-\cos^2 \theta_{XY})(1- \cos^2 \theta_{YZ}) \cos^2 \phi_{Y}$ must be rational. Again, since $\cos \theta_{XY}$ and $\cos \theta_{YZ}$ are both rational, $\cos^2 \phi_{Y}$ and hence $\cos 2 \phi_{Y}$ must be rational. But this is impossible since $\phi_Y$ is itself rational and $\phi_Y$ is not a multiple of $45^\circ$. Hence $\cos \theta_{XZ}$ must be irrational. 
\end{proof}

The Impossible Triangle Corollary is vital for explaining the notion of non-commutativity in RaQM \cite{Palmer:2022}. Consider a particle with spin prepared (with Stern-Gerlach device SG0) relative to some exact orientation $\bold X$. It is passed through Stern-Gerlach device SG1 with exact orientation $\bold Y$. The spin-up output beam of SG1 is passed through Stern-Gerlach device SG2 with exact orientation $\bold Z$. By RaQM, $\bold X \cdot \bold Y$ and $\bold Y \cdot \bold Z$ must be rational. In a run of the experiment, a detector in one output channel of SG2 will register a particle. Consider a hypothetical counterfactual experiment on the same particle (same $\lambda$) where SG1 and SG2 are swapped (commuted). The measurement outcome from this hypothetical experiment is undefined: by the Impossible Triangle Corollary, if  $\bold X \cdot \bold Y$ and $\bold Y \cdot \bold Z$ are rational, then $\bold X \cdot \bold Z$ is not.

We note in passing that it is straightforwardly shown \cite{Palmer:2020} that the ensemble representation of the single qubit state in RaQM satisfies a uncertainty principle relationship, i.e.
\be
\Delta S_x \Delta S_y  \ge \frac{\hbar}{2} \langle S_z \rangle
\ee
Here $\Delta S_x$ and $\Delta S_y$ are associated with standard deviations of bit-strings, and $\langle S_z \rangle$ denotes a bit-string ensemble mean. 

\subsubsection{Multiple Qubits}

Bell's inequality is based on measurements of pairs of entangled particles prepared in the quantum mechanical singlet state
\be
\label{singlet}
|\psi>=\frac{1}{\sqrt 2} [ |\bold x +\rangle |\bold y -\rangle-|\bold x -\rangle|\bold y +\rangle]
\ee
with correlations
\be
C(\bold x, \bold y)= \langle \psi | (\boldsymbol \sigma \cdot \bold x)\;(\boldsymbol \sigma \cdot \bold y) | \psi \rangle = - \bold x \cdot \bold y
\ee
where $\boldsymbol \sigma$ denote Pauli matrices. 

In RaQM, an $n$-qubit system is represented by a set of $n$ bit strings. Hence in RaQM (\ref{singlet}) is represented by two correlated bit strings (each with equal numbers of $+1$s and $-1$s). The bits represent measurement outputs from members of an ensemble, defined by deterministic laws. Corresponding to any one ensemble member (which can be labelled by $\lambda$) the exact measurement settings corresponding to $\bold x$, $\bold y$ must satisfy 
\be
\label{ang1}
\bold X \cdot \bold Y \in \mathbb Q
\ee
i.e., $\cos \theta_{XY}$ is rational. Keeping $X$ fixed and perturbing $Y \mapsto Y'$, then $X$ and $Y'$ are also permissible exact settings if $\bold X \cdot \bold Y'$ is rational, and the angle $\phi$ subtended at $X$ between the two great circles ${XY}$ and ${XY'}$ satisfies
\be
\label{ang2}
\frac{\phi}{2\pi} \in \mathbb Q
\ee
Similar rational constraints apply, fixing $Y$ and perturbing $X \mapsto X'$.

In any small neighbourhood of some point which does not satisfy the rationality conditions (\ref{ang1}) and (\ref{ang2}), there will, for large enough $p$, exist points which do satisfy the rational conditions. Hence it will be impossible to violate the rationality conditions in any coarse-graining no matter how fine, for large enough $p$. 

\subsection{The Invariant Set Postulate}

A clear disadvantage of discretised Hilbert space is that the sum of two discretised Hilbert vectors is no longer guaranteed to be a Hilbert vector - indeed typically it is not by Niven's theorem. However, if discretised Hilbert vectors represent symbolic strings describing ensembles of deterministic worlds, one can speculate that arithmetic closure exists at some deeper deterministic level. But what type of deterministic system would be consistent with the rationality constraints described above? The discussion in Section \ref{andromeda} suggests a chaotic model may be relevant. However, we need something in addition to mere chaos to account for the rationality constraints, we need a chaotic system evolving on an invariant subset of state space. 

Consider, for example, the chaotic model 
\begin{align}
\frac{dx_L}{dt}&= \sigma (y_L-x_L)\nonumber \\
\frac{dy_L}{dt}&= x_L(\rho-z_L) -y_L \nonumber \\
\frac{dz_L}{dt}&= x_Ly_L- \beta z_L
\end{align}
that Lorenz \cite{Lorenz:1963} discovered in his quest to understand the deterministic non-periodicity of weather. These equations describe a classical dynamical system and can be integrated from any triple of initial states $(x_L, y_L, z_L)$ at $t=0$. However, at $t=\infty$ this model has an emergent non-classical property: all states lie on a measure-zero, fractionally dimensioned, dynamically invariant subset $\mathscr I_L$ of state space, known as the Lorenz attractor. That $\mathscr I_L$ is fractionally dimensioned implies that $\mathscr I_L$ has a non-Euclidean, fractal geometry. That $\mathscr I_L$ is an invariant set implies that if a point lies on $\mathscr I_L$, its future evolution will continue to lie on $\mathscr I_L$ for all time, and its past evolution has lain on $\mathscr I_L$ for all time. That $\mathscr I_L$ has measure zero implies that the probability that a randomly chosen point belongs to $\mathscr I_L$ is equal to zero (random with respect to the uniform measure on the Euclidean state space $\mathbb R^3$ spanned by $(x_L,y_L, z_L)$). This is consistent with the notion that the rationals describe a set of measure zero in the continuum field of real numbers. Conversely, associated with $\mathscr I_L$ is a fractal invariant measure $\rho_L$ (a Haar-type measure). Points which do not lie on $\mathscr I_L$ have $\rho_L=0$. Such points define states which are inconsistent with a non-classical dynamical system where all states lie on $\mathscr I_L$ by definition. Such a system is non-classical because $\mathscr I_L$ is a non-computable subset of state space \cite{Blum} \cite{Dube:1993} - non-computability being a post-quantum discovery of 20th Century mathematics . 

Now suppose we are given some timeseries from the output of the Lorenz equations in this non-classical limit where the system is evolving on $\mathscr I_L$.  No matter how long is the timeseries, we cannot estimate statistical quantities such as correlations or conditional frequencies more accurately than the accuracy to which the timeseries has been outputted (e.g. the number of significant figures of output variables). That is to say, we must treat all estimates of frequency (and correlation) as functions of coarse-grained variables 
\be
\bar x_L = \int_{V_\epsilon} x_L \rho_L(x_L) dx_L; \ \ \bar y_L = \int_{V_\epsilon} y_L \rho_L(y_L) dy_L; \ \ \bar z_L = \int_{V_\epsilon} z_L \rho_L(z_L)dz_L
\ee
defined from non-zero $\epsilon$ balls of volume $V_\epsilon$  in state space. A key property of $\mathscr I_L$ is that no matter how small is $V_\epsilon$, as long as it is non-zero, it is undecidable as to whether a point inside $V_\epsilon$ has $\rho_L=0$. However, by the fractal nature of $\mathscr I_L$, we know such points exist. 

The results from Section \ref{andromeda} suggests the universe itself be considered a chaotic system. Consistent with the discussion above, we assume the universe is a deterministic chaotic dynamical system evolving precisely on its fractal invariant set $\mathscr I_U$, with invariant measure $\rho_U$. This is referred to as the invariant set postulate \cite{Palmer:2009a}. States $q$ which do not lie on $\mathscr I_U$ must be assigned a measure and hence a probability $\rho_U(q)=0$. It is worth noting that geometric properties of a system's invariant set (e.g. its non-integer dimension) provides a relativistically invariant description of chaos, in contrast with positivity of Lyapunov exponents \cite{Cornish:1997}. 

In number theory, Ostrowsky's theorem states that there are only two inequivalent norm-induced metrics on the rational numbers $\mathbb Q$: the Euclidean metric and the $p$-adic metric \cite{Katok}. It is well known that the set of $p$-adic integers is homeomorphic to a Cantor Set with $p$ iterated pieces \cite{Robert}.  We can therefore suppose that states in discretised Hilbert Space represent ensembles of trajectories on a fractal invariant set at some level of fractal iteration, where each trajectory at one level of iteration is associated with an ensemble of $p$ trajectories at the next level of fractal iteration. In this way, a deterministic state which does not satisfy the rationality constraints of RaQM corresponds to a state of the world which does not lie on the invariant set $\mathscr I_U$. In this picture, the measurement process corresponds to a jump from one fractal iteration of $\mathscr I_U$ to the next - that is to say a `quantum jump' describes an increment in fractal iteration. This suggests a picture of the evolution of time similar to a fractal zoom \cite{Susskind:2012}.

Although in a true fractal, the depth of iteration is infinite, the ideas expressed here continue to hold if the depth of iteration of $I_U$ is in fact finite. An invariant set with finite depth of iteration corresponds to a strictly periodic limit cycle. Computational representations of fractal attractors are in fact periodic limit cycles. We will assume below that $I_U$ is in fact a periodic limit cycle. 

As discussed above, although deterministic, such models are not classical. Classical models are associated with deterministic initial conditions and dynamical evolution equations expressed in terms of differential (or finite difference) equations on the reals or complex numbers. Typically one can vary initial conditions as one likes, and perturbed initial conditions can typically be integrated from the evolution equations without issue. In classical models, the ontology of states does not depend on their lying on invariant sets, nor on their having rational-number characteristics. This has consequences for our understanding of free will, discussed in Section \ref{other}. 

Notice that the results discussed here do not depend on how large is $p$, as long as it is not infinite. Moreover, by writing $\epsilon \sim 1/p$, it can be seen that violations to the rationality constraints in RaQM can be completely eliminated by coarse graining over $\epsilon$ disks, no matter how small is $\epsilon$.

\section{Bell's Theorem}
\label{belltheorem}

\subsection {The Bell (1964) Inequality}

In this subsection, we focus on the original Bell inequality \cite{Bell:1964}
\be
\label{bell}
 |C(\bold x_1, \bold x_2)-C(\bold x_1, \bold x_3)| \le 1+C(\bold x_2, \bold x_3)
\ee
For some specific run in an experiment to test (\ref{bell}), suppose Alice chooses the nominal orientation $\bold x_1$ and Bob the nominal orientation $\bold x_2$. In keeping with the discussion above, $\lambda$, $\bold x_1$ and $\bold x_2$ fix a pair of exact measurement settings $\bold X_{1}=\bold E(\lambda, \bold x_1)$ and $\bold X_{2}=\bold E(\lambda, \bold x_2)$. To satisfy (\ref{ang1}), $\bold X_{1} \cdot \bold X_{2} \in \mathbb Q$. Unlike the nominal settings, Alice and Bob have no control over exact settings, hence have no \emph{choice} as to whether this rationality constraint is satisfied. There are many points in the two $\epsilon$ disks for $\bold x_1$ and $\bold x_2$ which satisfy this constraint (see Fig \ref{fig1}). 

\begin{figure}
\centering
\includegraphics[scale=0.5]{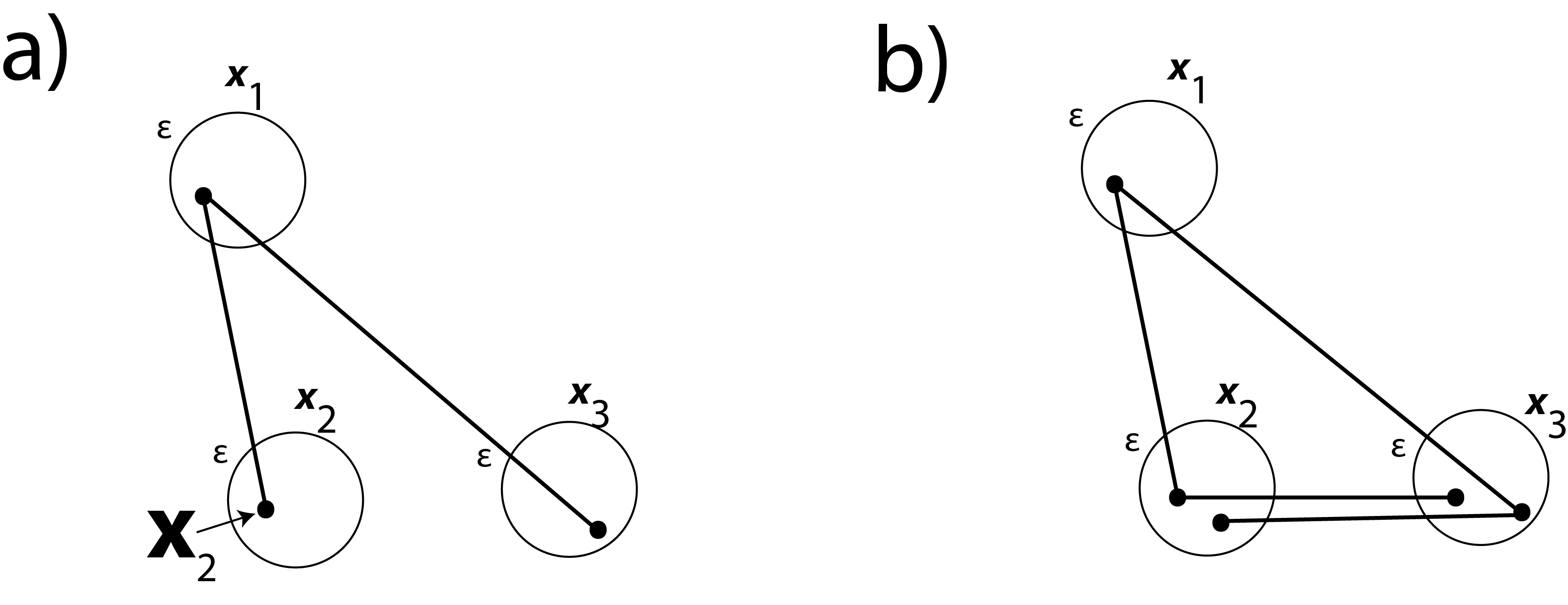}
\caption{\emph{a) The three circles correspond to $\epsilon$ disks on the unit sphere associated with the three nominal measurement settings $\bold x_1$, $\bold x_2$ and $\bold x_3$ in Bell's 1964 inequality. The straight lines represent great circles on the unit sphere whose cosine of angular distance is rational. By the Impossible Triangle Corollary, it is impossible for the cosine of all three angular distances (i.e. $\bold X_1 \cdot \bold X_2$,  $\bold X_1 \cdot \bold X_3$ and $\bold X_2 \cdot \bold X_3$) to be rational. Because of this, (\ref{statind}) is violated non-conspiratorially. b) In a more conventional model, there is no requirement for the exact settings to be held fixed when comparing real and hypothetical worlds with the same hidden variables and same nominal settings. In this model it is always possible to satisfy (\ref{statind}) and hence the violation of Bell inequalities must imply violation of local realism or some grotesque conspiracy.}}
\label{fig1}
\end{figure}

In order that a putative hidden-variable theory satisfies (\ref{bell}), it is necessary that, in addition to the real-world run where the particles were measured with nominal settings $(\bold x_1, \bold x_2)$, the same particles (same $\lambda$) could have been measured with nominal settings $(\bold x_1, \bold x_3)$ and $(\bold x_2, \bold x_3)$ with definite outcomes $\pm 1$. We will consider these hypothetical worlds in turn. 

For the first, keeping $\lambda$ fixed, Alice continues to choose the nominal setting $\bold x_1$, whilst Bob hypothetically chooses the nominal setting $\bold x_3$. Since $\bold X_1=\bold E(\lambda, \bold x_1)$, then keeping $\lambda$ and $\bold x_1$ fixed, $\bold X_1$ is fixed. By contrast, keeping $\lambda$ fixed but transforming $\bold x_2$ to $\bold x_3$ implies a hypothetical change in Bob's exact setting from $\bold X_2$ to some (unknown) $\bold X_3$. This transformation is consistent with RaQM as long as the rationality conditions (\ref{ang1}) and (\ref{ang2}) hold. It is therefore necessary that $\bold X_1 \cdot \bold X_3 \in \mathbb Q$ and that $\phi/2\pi \in \mathbb Q$ where $\phi$ denotes the angle between the great circles $X_1 X_2$ and $X_1 X_3$ at the point $X_1$ of intersection. There are plenty of exact settings $X_3$ in the $\epsilon$ neighbourhood for $\bold x_3$ for which these rationality constraints are satisfied. 

But now, c.f. the third term in (\ref{bell}), we consider the hypothetical world where, keeping $\lambda$ fixed, Alice chooses the nominal direction $\bold x_2$ and Bob $\bold x_3$. With $\bold X_3= \bold E(\lambda, \bold x_3)$, $\bold X_3$ is fixed by its value for the first hypothetical world. We now invoke the key property of the singlet state: if the measurement outcome was $+1$ (say) when Bob's exact setting was $\bold X_2$, then the measurement outcome will be $-1$ in a hypothetical world where Alice's exact setting was $\bold X_2$. Hence, considering the exact settings corresponding to the third term in (\ref{bell}), $\bold X_2$ and $\bold X_3$ must both be held fixed at their previously determined values. However, appealing to the Impossible Triangle Corollary, if $\bold X_1 \cdot \bold X_2$ and $\bold X_1 \cdot \bold X_3$ are rational and if $\phi$ is rational and not \emph{precisely} a multiple of $45^\circ$ (gravitational waves from Andromeda will help ensure that), then $\bold X_2 \cdot \bold X_3$ cannot be rational. See Fig \ref{fig1}. 

In essence, for each run in a Bell experiment, one of the two counterfactual runs is inconsistent with the rationality constraints of the hidden variable model and therefore must be assigned a probability $\rho=0$. Put another way
\be
\label{three}
\rho(\lambda | \bold x_1, \bold x_2) \times \rho(\lambda | \bold x_1, \bold x_3) \times \rho(\lambda | \bold x_2, \bold x_3)= 0
\ee
where, to emphasise, all orientations in (\ref{three}) are nominal. If one configuration occurs in reality (so that its probability is not identically zero) then (\ref{three}) implies that (\ref{statind}) is violated. 

\subsection{The CHSH Inequality}

This is a straightforward extension of the argument above though we no longer use the singlet property that, with the same exact settings, Alice and Bob must have opposite measurement outputs. Here $x$, $x'$, $y$ and $y'$ (with $x'=1-x$, $y'=1-y$ as before) denote four $\epsilon$ disks (i.e. nominal settings) on the 2-sphere of orientations. In a given run where Alice chooses $x$ and Bob $y$, we write the corresponding exact settings as $\bold X$, $\bold Y$ where  
\be
\label{rat1}
\bold X \cdot \bold Y \in \mathbb Q
\ee 
In order that a putative hidden-variable theory satisfies (\ref{be}), it is necessary that, in addition to the real-world run, the same particles (with the same $\lambda$) could have been measured with nominal settings $xy'$, $x'y$ and $x'y'$, with definite outcomes $\pm 1$.  

\begin{figure}
\centering
\includegraphics[scale=0.6]{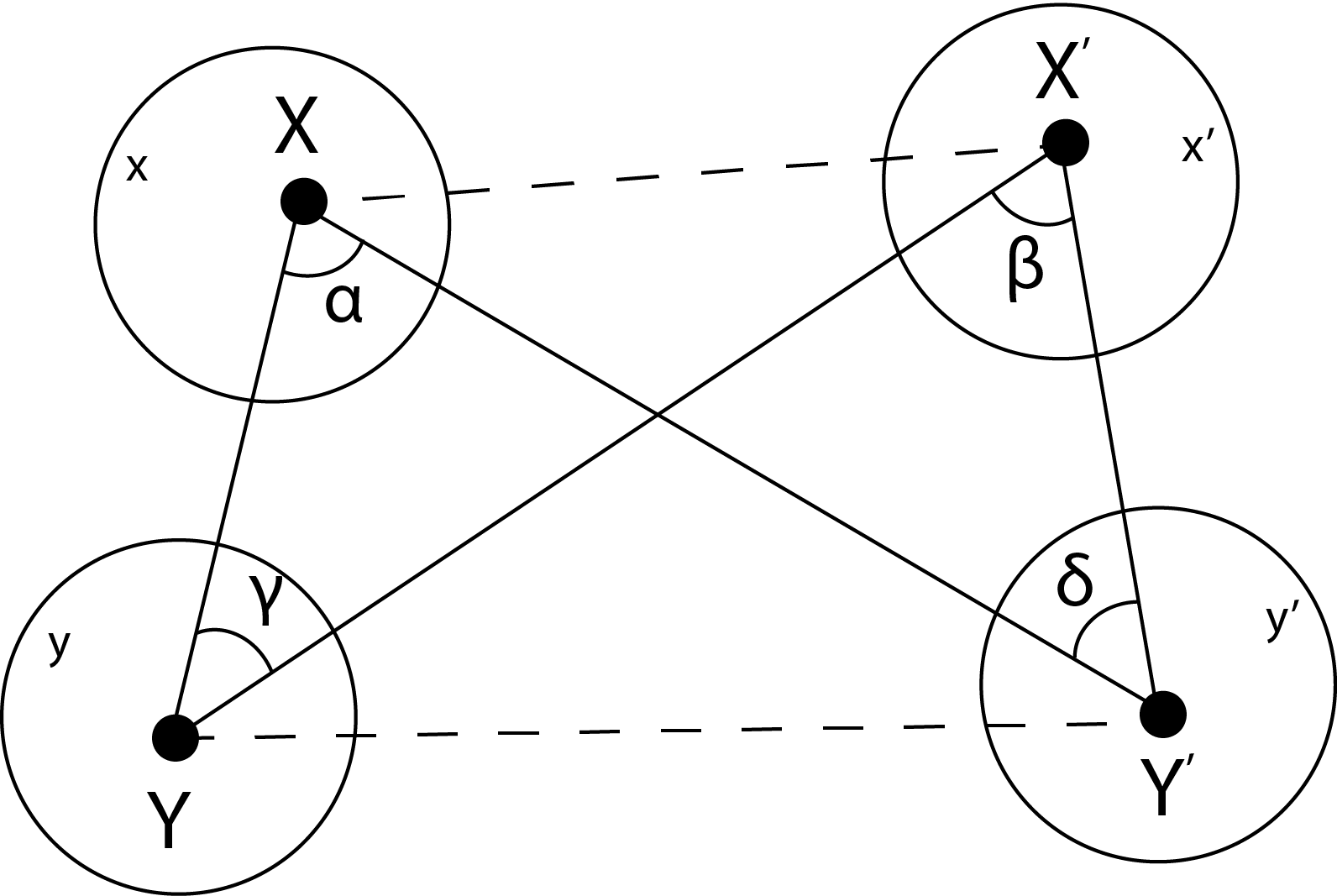}
\caption{\emph{Illustrating the CHSH experiment where $x$, $x'$, $y$ and $y'$ denote $\epsilon$-disks associated with nominal measurement settings, under the control of the experimenters.  $X$, $X'$, $Y$ and $Y'$ are points on the unit 2-sphere corresponding to exact measurement settings, and straight lines correspond to great circles joining these points. By the relationship $\bold X= \bold E(\lambda, \bold x)$ and $\bold Y = \bold E(\lambda, \bold y)$, fixing $\lambda$ and a nominal measurement setting fixes an exact measurement setting.}As discussed in the text, if Alice and Bob chose the nominal settings $xy$, then, by the Impossible Triangle Corollary with $\lambda$ fixed, the exact settings associated with one of the three counterfactual nominal settings $xy'$, $x'y$ or $x'y'$ cannot satisfy the rationality conditions of RaQM.}
\label{fig2}
\end{figure}

As before, with $\bold X=\bold E(\lambda, \bold x)$,  $\bold X'=\bold E(\lambda, \bold x')$, $\bold Y=\bold E(\lambda, \bold y)$, $\bold Y'=\bold E(\lambda, \bold y')$, we seek two exact settings $\bold X'$, $\bold Y'$ such that:
\be
\label{rat2}
\bold X \cdot \bold Y' \in \mathbb Q;\ \ \bold X' \cdot \bold Y \in \mathbb Q;\ \ \bold X' \cdot \bold Y' \in \mathbb Q
\ee
and, see Fig \ref{fig2}, 
\be
\label{rat3}
\frac{\alpha}{2\pi} \in \mathbb Q; \ \ \frac{\beta}{2\pi} \in \mathbb Q; \ \ \frac{\gamma}{2\pi} \in \mathbb Q; \ \ \frac{\delta}{2\pi} \in \mathbb Q.
\ee
By repeated application of the Impossible Triangle Corollary, we now show it is impossible to satisfy (\ref{rat2}) and (\ref{rat3}). 

We do this again by contradiction. Suppose that (\ref{rat2}) and (\ref{rat3}) are satisfied and consider the two triangles $\triangle XX'Y'$ and $\triangle XX'Y$. By the cosine rule for spherical triangles for each of the two triangles
\begin{align}
\label{cos1}
\cos \theta_{XX'}=& \cos \theta_{XY} \cos \theta_{X'Y}+\sin \theta_{XY} \sin \theta_{X'Y} \cos \gamma \nonumber \\
\cos \theta_{XX'}= &\cos \theta_{X'Y'} \cos \theta_{XY'}+\sin \theta_{X'Y'} \sin \theta_{XY'} \cos \delta
\end{align}
Subtracting these equations then 
\be
A=\sin \theta_{XY} \sin \theta_{X'Y} \cos \gamma - \sin \theta_{X'Y'} \sin \theta_{XY'} \cos \delta
\ee
must be rational. Writing $A_1=\sin \theta_{XY} \sin \theta_{X'Y} \cos \gamma$, $A_2=\sin \theta_{X'Y'} \sin \theta_{XY'} \cos \delta$, then we can write
\be
A_1^2= r_1(1+\cos 2 \gamma); \ A_2^2=r_2(1+\cos 2\delta)
\ee
where $r_1, r_2$ are rational. However, by Niven's Theorem, providing $\gamma$ and $\delta$ are not precise multiples of $45^\circ$, $A_1^2$ and $A_2^2$ must be irrational, hence $A_1$ and $A_2$ must be irrational. Moreover they must be independently irrational since $\gamma$ and $\delta$ can be varied independently of one another. Hence generically $A_1-A_2=A$ must be irrational which is the contradiction we are looking for. 

This in turn leads to the following general conclusion. In the situation where Alice chose $x$ and Bob $y$, then, keeping the particles' hidden variables fixed, at least one of the three counterfactual choices must violate the rationality conditions: 1) Alice and Bob chose $x$ and $y'$, Alice and Bob chose $x'$ and $y$, or Alice and Bob chose $x'$ and $y'$. Similar to (\ref{three})
\be
\label{four}
\rho(\lambda | xy) \times \rho(\lambda | xy') \times \rho(\lambda | x'y) \times \rho(\lambda, x'y') = 0
\ee
If $\rho(\lambda | xy) \ne 0$, i.e. one configuration occurs in reality then (\ref{four}) implies that (\ref{statind}) is violated. As with Bell's 1964 inequality, the Impossible Triangle Corollary implies that (\ref{statind}) is violated without conspiracy.

It can be noted that since RaQM is based on an arbitrarily fine discretisation of complex Hilbert Space, by letting $p$ be sufficiently large, RaQM must violate Bell's inequality as closely we like to the quantum mechanical violation of Bell's inequality. 

\subsection{Local Causality}
\label{lc}

Below we consider a locally causal violation of the rationality constraints (\ref{rat1}) and (\ref{rat2}). 

Fig \ref{fig3} illustrates a space-time diagram where a photonic source emits two entangled photons. $\mathcal S$ is a spacelike hypersurface through the event where the particles were emitted by the source. The photons are measured by Alice and Bob's detectors with nominal settings $x$ and $y$ and exact settings $X$ and $Y$. The nominal settings are determined by two PRNGs shown in the figure. By the discussion above, the particle's hidden variables $\lambda$, together with the parities $P_x$ and $P_y$, determine these exact settings. We suppose that, consistent with the counterfactual violation of  (\ref{rat1}) and (\ref{rat2}), $\bold X \cdot \bold Y \in \mathbb Q$ but $\bold X \cdot \bold Y' \notin \mathbb Q$ and $\bold X' \cdot \bold Y' \notin \mathbb Q$

We divide $\lambda$ into two components, $\lambda_A$ corresponding to information in the past light cone of Alice's measurement event, and $\lambda_B$ corresponding to information in the past light cone of Bob's measurement event. We write Alice and Bob's measurement outcomes ($=\pm1$) in the form:
\be
\label{SASB}
S_A=S(\lambda_A, \lambda_B, x, y); \ \ S_B=S(\lambda_A, \lambda_B, x, y)
\ee
Whatever the values of $S_A$ and $S_B$ in (\ref{SASB}), it is \emph{never} the case that
\begin{align}
\label{lc1}
S(\lambda_A, \lambda_B, x, y')&=-S_A; \ \ S(\lambda_A, \lambda_B, x', y)=-S_B 
\nonumber \\
S(\lambda_A, \lambda'_B, x, y)&=-S_A; \ \ S(\lambda'_A, \lambda_B, x, y)\ =-S_B
\end{align}
where $\lambda'_A$, $\lambda'_B$ describe perturbed data on $\mathcal S$ (see Fig \ref{fig3}). That is to say, a hypothetical intervention \emph{in} space-time which alters either Bob's PRNG input or the hidden variable $\lambda_B$ in the past light cone of Bob's measurement event, keeping $P_x$ and $\lambda_A$ fixed, and changes Alice's measurement outcome, violates the rationality conditions (\ref{rat1}) and (\ref{rat2}). In simple words this implies that Alice's measurement outcome does not depend on Bob's measurement settings - the essence of locality.  The space-time events at $B$ are irrelevant for determining the event $A$ because any intervention which changes $B$ keeping the data on $\mathcal S_A$ fixed is inconsistent with the putative model and hence does not define a putative event in space time. But this is precisely Bell's characterisation of local causality in La Nouvelle Cuisine  - see Fig 4 in \cite{Bell:2004}. Put another way, the event $A$ is entirely determined by data on $\mathcal S_A$, the subset of $\mathcal S$ contained on and in the past lightcone of $A$ - the essence of relativistic causality. 

\begin{figure}
\centering
\includegraphics[scale=0.5]{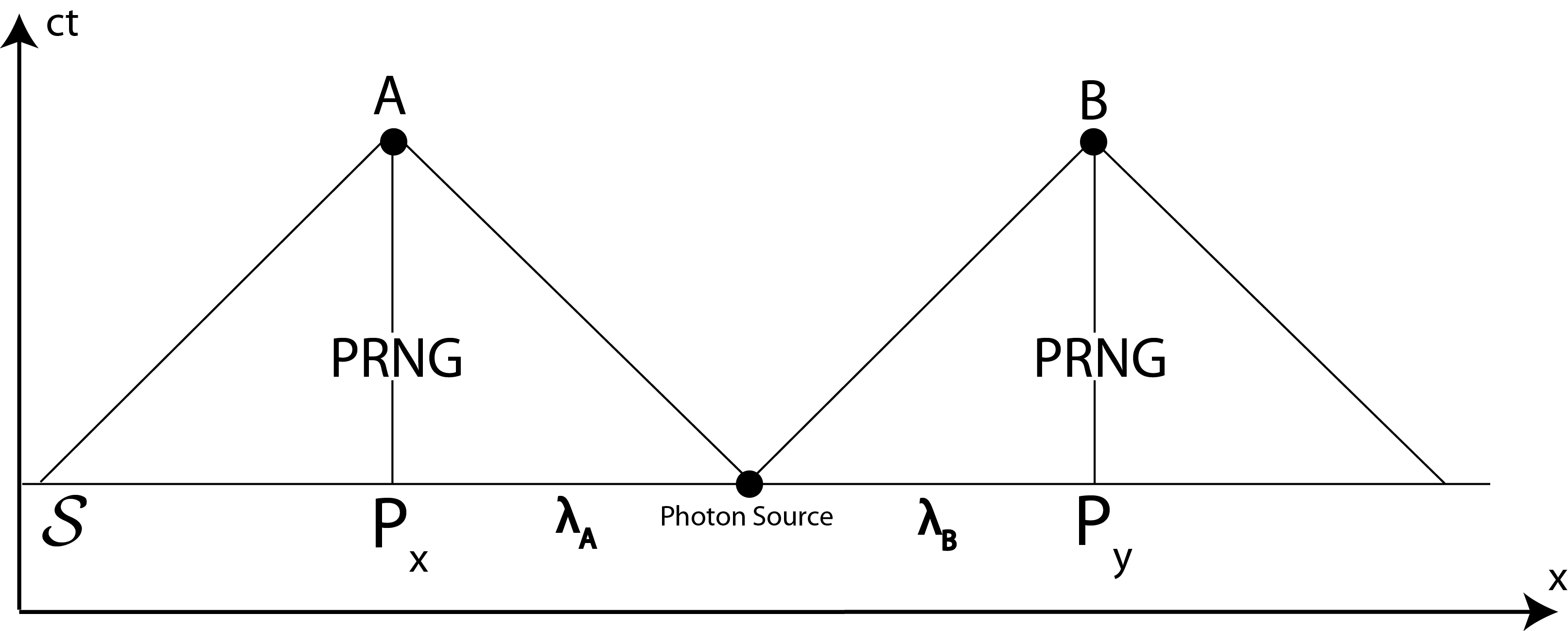}
\caption{\emph{A space-time diagram illustrating the locally causal nature of the proposed superdeterministic model based on RaQM. Suppose Alice's measurement outcome is $+1$. This measurement outcome could have been $-1$ if $P_x$ (the parity of the millionth digit) were different or if $\lambda_A$ were different. However, local causality demands that Alice's measurement outcome could not have been $-1$ if $P_y$ were different or $\lambda_B$ was different, keeping $P_x$ and $\lambda_A$ fixed. Importantly, this does not exclude the possibility that Alice's measurement outcome would be inconsistent with the laws of physics, and hence is undefined, if $P_y$ had been different or $\lambda_B$ had been different, keeping $P_x$ and $\lambda_A$ fixed. According to RaQM and the invariant set postulate, such undefinedness arises because the corresponding Pearlean interventions are inconsistent with rationality constraints and take the universe off its invariant set. Because of this, the event $A$, whilst determined by data on and inside the past light cone of $A$, does not depend on space-time events at or near $B$, consistent with Bell's notion of locality in La Nouvelle Cuisine.}}
\label{fig3}
\end{figure}

\section{Objections to Superdeterminism}
\label{other}
Below we address some of the objections that have been raised against superdeterminism with the RaQM/invariant set model in mind. 

\subsection{The Drug Trial}
\label{drug}

Each human in a drug trial is unique; having unique DNA for example. However, the characteristics used to allocate people randomly to the active drug or placebo groups are based on coarse-grain attributes. Are they young or old? Are they male or female? Are they black or white? We typically assume the two groups contain equal numbers of such coarse-grain attributes. In a conspiratorial drug trial, the selection process is manipulated so that the two groups do not contain equal numbers of coarse-grain attributes. 

Although RaQM violates (\ref{statind}), it does not violate any version of (\ref{statind}) where coarse-grained hidden variables are used in place of hidden variables. For large enough $p$, it is always possible to find counterfactual worlds which satisfy the rationality constraints in a coarse-grained volume $V_\epsilon$ of state space, no matter how small is $\epsilon$. That is to say, in any sufficiently-large ensemble of individual runs, 
\be
\label{freqind}
\rho_f(\bar \lambda | xy)=\rho_f(\bar \lambda | x'y)=\rho_f(\bar \lambda | xy')=\rho_f(\bar \lambda | x'y')=\rho_f(\bar \lambda)
\ee
where $\bar \lambda$ denotes a coarse-grained value of $\lambda$ and $\rho_f$ denotes a frequency of occurrence. In this sense, (\ref{freqind}) does not imply (\ref{statind}) and the drug trial conspiracy is an irrelevance to RaQM.

\subsection{Fine Tuning}

The fine-tuned objection (e.g. \cite{Baas:2020}) rests on the notion that superdeterminism appears to require some special, atypical, initial conditions. Perhaps one might view an initial state lying on a fractal attractor as special and atypical - after all a seemingly tiny perturbation (changing $P_x$ keeping $\lambda$ fixed) can take the state of the universe off its invariant set $\mathscr I_U$. Although the Euclidean metric accurately describes distances in space-time, the $p$-adic metric is more natural in describing distances in state space when the inherent geometry of state space is fractal \cite{Katok}. From the perspective of the $p$-adic metric, a fractal invariant set is not fine-tuned: a perturbation which takes a point off $\mathscr I_U$ is a large perturbation (of magnitude at least $p$), even though it may appear very small from a Euclidean perspective.  Conversely, perturbations which map points of $\mathscr I_U$ to points of $\mathscr I_U$ can be considered small amplitude perturbations. 

Similarly, we must ask with respect to what measure are states on $\mathscr I_U$ deemed atypical. Although states on $\mathscr I_U$ are atypical with respect to a uniform measure on the Euclidean space in which $\mathscr I_U$ is embedded, they are manifestly typical with respect to the invariant measure of $\mathscr I_U$ \cite{Hance2022Supermeasured}.   

In claiming that a theory is fine tuned, one should first ask with respect to which metric/measure is the tuning deemed fine -  and then ask whether this the natural metric/measure to assess fineness. 

\subsection{Singular Limits and Aaronson's Challenge}
\label{aaronson}

Aaronson's challenge (see the Introduction) raises a more general question: what is the relationship between a successor theory of physics and its predecessor theory? There is a subtle but important relationship brought out explicitly by Michael Berry \cite{Berry:2002} that is of relevance here. 

Typically an old theory is a singular limit of a new theory, and not a smooth limit, as a parameter of the new theory is set equal to infinity or zero. A singular limit is one where some characteristic of the theory changes discontinuously \emph{at} the limit, and not continuously \emph{as} the limit is approached. Berry cites as examples the old theory of ray optics is explained from Maxwell theory, or the old theory of thermodynamics is explained from statistical mechanics. His claim is that old theories of physics are \emph{typically} singular limits of new theories.

If quantum theory is a forerunner of some successor superdeterministic theory, and Berry's observation is correct, quantum mechanics is likely to be a singular limit of that superdeterministic theory. Here the state space of quantum theory arises at the limit $p=\infty$ of RaQM, but not before. For any finite $p$, no matter how big, the incompleteness property that led to the violation of (\ref{statind}) holds. However, it does not hold at $p=\infty$. From this point of view, quantum mechanics is indeed a singular limit of RaQM's discrete Hilbert space, at $p=\infty$. It is interesting to note that pure mathematicians often append the real numbers to sets (`adeles') of $p$-adic numbers, at $p=\infty$. However, the properties of $p$-adic numbers are quite different to those of the reals for any finite $p$ no matter how big. Here the real-number continuum is the singular limit of the $p$-adics at $p=\infty$. The relationship between QM and RaQM is very similar.

In physics, it is commonplace to solve differential equations numerically, i.e. to treat discretisations as approximations of some continuum exact equations, such that when the discretisation is fine enough, the numerical results are as close as we require to the exact continuum solution. This is \emph{not} a good analogy here. A better analogy is analytic number theory, considered as an approximation to say the exact theory of prime numbers. If one is interested in properties of primes for large primes, treating $p$ as if it were a continuum variable can provide excellent results. However, here the continuum limit is the approximation and not the exact theory. 

Contrary to Aaronson's statements, the singular relationship between a superdeterministic theory and quantum mechanics is exactly as one would expect from the history of science. 

\subsection{Free Will}

Superdeterminism is often criticised as denying experimenter free will. Nobel Laureate Anton Zeilinger put it like this \cite{Zeilinger}:
\begin{quote}
We always implicitly assume the freedom of the experimentalist... This fundamental assumption is essential to doing science. If this were not true, then, I suggest, it would make no sense at all to ask nature questions in an experiment, since then nature could determine what our questions are, and that could guide our questions such that we arrive at a false picture of nature.
\end{quote}
A clear problem here is that the notion of free will is poorly understood \cite{Kane} and therefore hard to define rigorously. It was for this reason that Bell introduced his PRNG gedanken experiment - to show that it was possible to discuss (\ref{statind}) and its potential violation without invoking free will. 

Nevertheless, to avoid the charge of conspiracy, an experimenter must be able to choose in a way which is indistinguishable from a random choice. For the present purposes we can think of this as being consistent with free will. For this reason, experimenters have found increasingly whimsical ways of choosing measurement settings - such as bits from a movie, or the wavelength of light from a distant quasar - in an attempt to mimic randomness.  

In Fig \ref{fig4}  we replace the two PRNGs in Fig \ref{fig3} by Alice and Bob's brains. It is well known that brains are low-power noisy systems \cite{RollsDeco} \cite{Palmer:2022b} where neurons can be stochastically triggered. In practice, the source of such stochasticity is thermal noise. However, such noise - arising from the collision of molecules - will have an irreducible component due to the Andromedan butterfly effect.  In such circumstances, as discussed in Section \ref{andromeda}, this can render the action of the brain non-computational. In particular, if we were to construct a model of Alice and Bob's brains which are driven by the subset $\lambda$ of data on $\mathcal S$ (or indeed any subset), this model will not provide reliable predictions of their brains' decisions. This surely provides evidence of our ability to choose in ways which are for all practical purposes random. 

\begin{figure}
\centering
\includegraphics[scale=0.5]{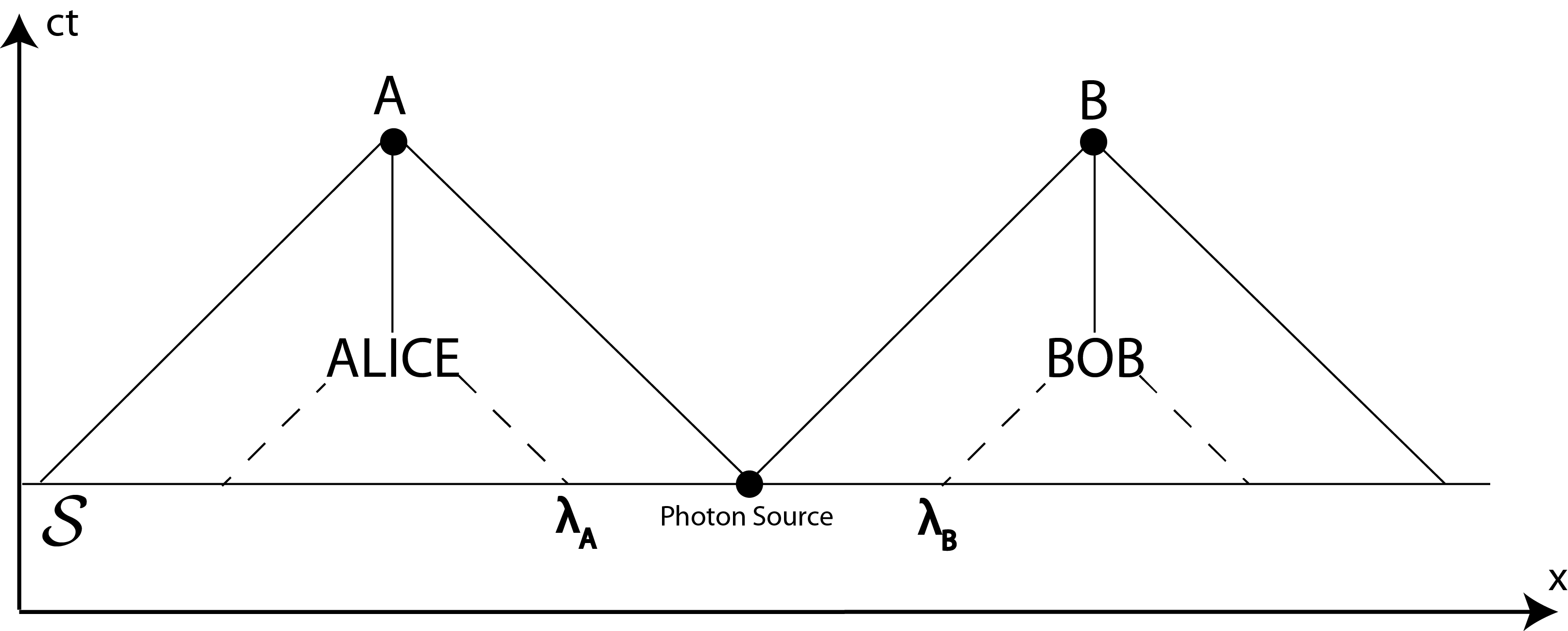}
\caption{\emph{Here the PRNGs of Fig \ref{fig3} are replaced with the experimenters' brains. By the Andromedan butterfly effect, we can assume that the outputs of these brains is sensitive to all data on $\mathcal S$ in ALICE and BOB's past light cones and not just $\lambda$. We can use this to assert that Alice and Bob's brains cannot be corrupted by the values of the hidden variables.}}
\label{fig4}
\end{figure}
 
However, the Invariant Set Postulate provides insights which may help shed new light on the age-old dilemma of free will. Consider two possible invariant sets $\mathcal I_U$ and $\mathcal I_{U'}$. Here, for given $\lambda$, we suppose $\mathcal I_U$ permits the settings $00$ and $11$ in a Bell experiment, whilst, keeping $\lambda$ fixed, $\mathcal I_{U'}$ permits the settings $01$ and $10$. These invariant sets differ (very slightly) in terms of their geometries, i.e. in terms of the (underlying deterministic) laws of physics. 

If Alice and Bob are free to choose their measurement settings, then, prior to their choosing, all observations of the universe must be consistent with the universe belonging to either of the $\mathcal I_U$ and $\mathcal I_{U'}$. However, once Alice and Bob have chosen, not only is one of $\mathcal I_U$ and $\mathcal I_{U'}$ consistent with available observations, the other becomes inconsistent with the supposed deterministic laws of physics. 

Suppose Alice and Bob chose $00$ and let $q$ denote a state of the universe prior to Alice and Bob choosing. The question then arises: was it always the case - even before Alice and Bob chose $00$ -  that $q \in \mathcal I_U$ and not $\mathcal I_{U'}$? Or did the very act of Alice and Bob choosing result in $q \in \mathcal I_U$ rather than $q \in \mathcal I_{U'}$? In thinking about these questions, it is important to distinguish determinism from pre-destination. In a conventional deterministic initial-value problem, initial conditions are specified independently from evolutionary laws and the evolved state is predestined from the initial state. By contrast, if  one thinks of the geometry of the invariant set as primitive, then the choices Alice and Bob make are no more `predestined' from $q$ than $q$ is predestined from their choices. Instead, all one can say is that, because of determinism, the earlier and later states must be dynamically consistent. Importantly, as a global state space geometry, $\mathcal I_U$ is consistent with what Adlam calls an `all at once' constraint \cite{Adlam:2022}. One could say that the geometric specification of $\mathcal I_U$, and hence whether $q \in \mathcal I_U$, depends as much on states on $\mathcal I_U$ to the future of $q$ as on states to the past of $q$. This future/past duality exists because neither the proposition $q \in \mathcal I_U$, nor Alice and Bob's choice given $q$, is computational (in our finite system it is computationally irreducible). 

Hence, is it both true that $q \in \mathcal I_U$ (rather than $\mathcal I_{U'}$) before Alice and Bob chose, and also true that the act of choosing required $q \in \mathcal I_U$ rather than $q \in \mathcal I_{U'}$. Importantly, as there is no notion of temporal causality in state space, it would be wrong to call this latter fact retrocausality. Simply, it is a consequence of $\mathcal I_U$ being an all-at-once constraint. This type of analysis helps explain the so-called delayed choice paradoxes in quantum physics. 

It is possible to conclude not only that experimenter choices are indeed freely made (Nobel Laureates can be assured that their brains are not being subverted), but that these choices can determine which states of the universe are consistent with the laws of physics and which not (surely a fitting role for Nobel Laureates). This has some significant implications for the role of intelligent life in the universe more generally, which the author will discuss elsewhere. 

\section{Experimental Tests}
\label{experiment}

A key result from this paper is that we will not be able to detect non-conspiratorial superdeterministic violations of (\ref{statind}) by studying frequencies of measurement outcomes in a Bell experiment. We must look for other ways of testing such theories. 

Of course, QM is exceptionally well tested and if a superdeterministic theory is to replace QM, it must clearly be consistent with results from all the experiments which support QM. Here RaQM has a free parameter $p$, which, if large enough, can replicate all existing experiments. This is because with large enough $p$, discretised Hilbert space is fine enough that it replicates to experimental accurace the probabilistic predictions of a theory based on continuum Hilbert space (and Born's Rule to interpret the squared modulus of a state as a probability - something automatically satisfied in RaQM). Conversely, however, if $p$ is some finite albeit large number, then in principle an experiment with free parameter $p_{exp}$ can study situations where $p_{exp} > p$ where there might be some departure between reality and QM \cite{HanceHossenfelder:2022a}.

One conceivable test of RaQM vs QM probes the finite amount of information that can be contained in the quantum state vector $|\psi\rangle$ (in RaQM $n$ qubits are represented by $n$ bit strings of length $p$). In RaQM, the finite information encoded in the quantum state vector will limit the power of a general purpose quantum computer, in the sense that RaQM predicts the exponential increase in quantum compute speed with qubit number for a range of quantum algorithms may generally max out at a finite number $m$ of qubits. 

The key question concerns the value of $m$. Could it be related to the number of collisions in the Andromeda butterfly effect? We will return to this elsewhere. 

\section{Conclusions}
\label{conclusions}

Attempts to develop models which violate (\ref{statind}) do not justify the derision from a number of researchers in the quantum foundations community over the years. Not least, Bell himself did not treat the possible violation of (\ref{statind}) with derision and accepted that seemingly reasonable ideas about the properties of physical randomisers might be wrong - for the purposes at hand. . 

A superdeterministic (and hence deterministic) model has been proposed which is not conspiratorial, is locally causal, does not deny experimenter free choice and is not fine tuned with respect to natural metrics and measures. The model, based on a discretisation of Hilbert space, is not a classical hidden-variable model, i.e. it derives its properties from post-quantum-theory mathematical science (particularly that of non-computability and computational irreducibility). By considering the continuum of complex Hilbert Space as a singular limit of a superdeterministic discretisation of complex Hilbert Space, Aaronson's challenge to superdeterminists, to show how a superdeterministic model might gloriously explain quantum mechanics, can be met. 

One of the most important conclusions of this paper is that we need to be extremely cautious when invoking the notion of an `intervention' in space time, at least in the context of fundamental physics. Such interventions form the bedrock of Pearl's causal inference modelling \cite{Pearl}, and causal inference has been used widely in the quantum foundations community to try to analyse the causal structure of quantum physics e.g. \cite{Spekkens} \cite{PriceWharton}. Here we distinguish between two types of intervention: one that is consistent with the laws of physics and one that is not. The effect of the former type of intervention, if it is initially contained within a localised region of space-time, must propagate causally in space-time, constrained by the Lorentzian metric of space time. By contrast, the latter type of intervention simply perturbs a state of the universe from a part of state space where the laws of physics hold, to a part of state space where the laws of physics do not hold. If this superdeterministic model is correct, theories of quantum physics based on causal inference models which adopt an uncritical acceptance of interventions will give misleading results. 

The results of this paper suggest that the way gravity interacts with matter may be central to understanding the reasons why the universe can be considered a holistic dynamical system evolving on an invariant set, and hence why Hilbert space should be discretised. This suggests that instead of looking for a quantum theory of gravity, we should instead be looking for a gravitational theory of the quantum \cite{Palmer:2012} \cite{Penrose:2014}. However, importantly, the results here suggest such a theory will not be found by probing smaller and smaller regions of space-time, ultimately the Planck scale. It will instead be found by incorporating into the fundamental laws of physics, the state-space geometry of the universe at its very largest scales \cite{Palmer:2022b}. Planck-scale discontinuities in space-time may instead be an emergent property of such (top-down) geometric laws of physics. 

In this regard, a recent proposal \cite{Oppenheim:2023} for synthesising quantum and gravitational physics describes gravity as a classical stochastic field. The latter is consistent with our discussion of the Andromedan butterfly effect. However, it is not consistent with our discussion of Bell's inequality. On the other hand, if one simply acknowledges that a stochastic gravitational field is a `for all practical purposes' representation of a chaotic system evolving on a fractal invariant set, then Oppenheim's model may become consistent with both the proposed superdeterministic violation of Bell's inequality and with realism and the relativistic notion of local causality.  

In the author's opinion, this is the real message - not non-locality, indeterminism or unreality - behind the violation of Bell's inequality. 

\section*{Acknowledgements}

My thanks to Emily Adlam, Jean Bricmont, Harvey Brown, Michael Hall, Jonte Hance, Inge Svein Helland, Sabine Hossenfelder, Tim Maudlin and Chris Timpson for helpful discussions and/or useful comments on an early draft of this paper. The author gratefully acknowledges the support of a Royal Society Research Professorship. 

\bibliography{mybibliography}

\begin{thebibliography}{10}

\bibitem{Aaronson}
S.~Aaronson.
\newblock On tardigrades, superdeterminism and the struggle for sanity.
\newblock https://scottaaronson.blog/?p=6215, 2022.

\bibitem{Adlam:2022}
E.~Adlam.
\newblock Two roads to retrocausality.
\newblock {\em Sythese}, 200, 2022.

\bibitem{Araujo:2019}
M.~Araujo.
\newblock Superdeterminism is unscientific.
\newblock
  https://mateusaraujo.info/2019/12/17/superdeterminism-is-unscientific/, 2019.

\bibitem{Baas:2020}
A.~Baas and B.~LeBihan.
\newblock What does the world look like according to superdeterminism?
\newblock {\em British Journal for the Philosophy of Science}, 73.4:555--572,
  2023.

\bibitem{Bell:1964}
J.S. Bell.
\newblock On the {E}instein-{P}odolsky-{R}osen paradox.
\newblock {\em Physics}, 1:195--200, 1964.

\bibitem{Bell:2004}
J.S. Bell.
\newblock {\em Speakable and unspeakable in quantum mechanics}.
\newblock Cambridge University Press, 2004.

\bibitem{Belletal:1985}
J.S. Bell, A.Simony, M.A.Horne, and J.F.Clauser.
\newblock An exchange on local beables.
\newblock {\em Dialectica}, 39:85--110, 1985.

\bibitem{Berry:1978}
M.~Berry.
\newblock Regular and irregular motion.
\newblock American Institute of Physics Conference Proceedings Number 46, 1985.

\bibitem{Berry:2002}
M~Berry.
\newblock Singular limits.
\newblock {\em Physics Today}, 55:10--11, 2002.

\bibitem{Blum}
L.~Blum, F.Cucker, M.Shub, and S.Smale.
\newblock {\em Complexity and Real Computation}.
\newblock Springer, 1997.

\bibitem{Buniy:2005}
R.~Buniy, S.~Hsu, and A.~Zee.
\newblock Is {H}ilbert space discrete.
\newblock {\em Phys.Lett.}, B630:68--72, 2005.

\bibitem{Buniy:2006}
R.~Buniy, S.~Hsu, and A.~Zee.
\newblock Discreteness and the origin of probability in quantum mechanics.
\newblock {\em arXiv:hep-th/0606062}, 2006.

\bibitem{Carroll:2023}
S.~Carroll.
\newblock Completely discretized, finite quantum mechanics.
\newblock {\em arXiv:2307.11927}, 2023.

\bibitem{Chen:2020}
Eddy~K. Chen.
\newblock Bell's theorem, quantum probabilities and superdeterminism.
\newblock {\em arXiv:2006.08609}, 2020.

\bibitem{Cornish:1997}
N.J. Cornish.
\newblock Fractals and symbolic dynamics as invariant descriptors of chaos in
  general relativity.
\newblock arXiv.org/gr-qc/9709036, 1997.

\bibitem{Dube:1993}
S.~Dube.
\newblock Undecidable problems in fractal geometry.
\newblock {\em Complex Systems}, 7:423--444, 1993.

\bibitem{Ellis:2018}
G.F.R. Ellis, K.A.Meissner, and H.Nicolai.
\newblock The physics of infinity.
\newblock {\em Nature}, 14:770--772, 2018.

\bibitem{Goldsteinetal:2011}
S.~Goldstein, T.~Norsen, D.V.Tausk, and N.Zanghi.
\newblock Bell's theorem.
\newblock http://dx.doi.org/10.4249/scholarpedia.8378, 2011.

\bibitem{Hall:2010}
M.J.W. Hall.
\newblock Local deterministic model of singlet state correlations based on
  relaxing measurement independence.
\newblock {\em Phys. Rev. Lett.}, 105:250404, 2011.

\bibitem{Hance2022Supermeasured}
Jonte~R. Hance, Sabine Hossenfelder, and Tim~N. Palmer.
\newblock Supermeasured: Violating bell-statistical independence without
  violating physical statistical independence.
\newblock {\em Foundations of Physics}, 52(4):81, Jul 2022.

\bibitem{HanceHossenfelder:2022a}
J.R. Hance and S.~Hossenfelder.
\newblock What does it take to solve the measurement problem?
\newblock {\em arXiv:2206.10445}, 2022.

\bibitem{HossenfelderPalmer}
Sabine Hossenfelder and Tim Palmer.
\newblock Rethinking superdeterminism.
\newblock {\em Frontiers in Physics}, 8:139, 2020.

\bibitem{Jahnel:2005}
J.~Jahnel.
\newblock When does the (co)-sine of a rational angle give a rational number?
\newblock arXiv:1006.2938, 2010.

\bibitem{Kane}
R.~Kane.
\newblock {\em Free Will}.
\newblock Blackwell, 2002.

\bibitem{Katok}
S.~Katok.
\newblock {\em p-adic Analysis compared with Real}.
\newblock American Mathematical Society, 2007.

\bibitem{Lorenz:1963}
E.N. Lorenz.
\newblock Deterministic nonperiodic flow.
\newblock {\em J.Atmos.Sci.}, 20:130--141, 1963.

\bibitem{Lorenz:1969}
E.N. Lorenz.
\newblock The predictability of a flow which possesses many scales of motion.
\newblock {\em Tellus}, 21:289--307, 1969.

\bibitem{Maudlin}
T.~Maudlin.
\newblock {\em Quantum non-locality and relativity}.
\newblock Wiley-Blackwell, 2011.

\bibitem{Maudlin:2023}
T.~Maudlin.
\newblock Tim maudlin and palmer: Fractal geometry, non-locality, bell.
\newblock https://www.youtube.com/watch?v=883R3JlZHXE, 2023.

\bibitem{Niven}
I.~Niven.
\newblock {\em Irrational Numbers}.
\newblock The Mathematical Association of America, 1956.

\bibitem{Oppenheim:2023}
Jonathan Oppenheim.
\newblock A postquantum theory of classical gravity?
\newblock {\em Phys. Rev. X}, 13:041040, Dec 2023.

\bibitem{Palmer:1995}
T.N. Palmer.
\newblock A local deterministic model of quantum spin measurement.
\newblock {\em Proc. Roy. Soc.}, A451:585--608, 1995.

\bibitem{Palmer:2009a}
T.N. Palmer.
\newblock The invariant set postulate: a new geometric framework for the
  foundations of quantum theory and the role played by gravity.
\newblock {\em Proc. Roy. Soc.}, A465:3165--3185, 2009.

\bibitem{Palmer:2012}
T.N. Palmer.
\newblock Quantum theory and the symbolic dynamics of invariant sets: Towards a
  gravitational theory of the quantum.
\newblock arXiv:1210.3940, 2012.

\bibitem{Palmer:2020}
T.N. Palmer.
\newblock Discretization of the {B}loch sphere, fractal invariant sets and
  {B}ell's theorem.
\newblock {\em Proc. Roy. Soc.}, https://doi.org/10.1098/rspa.2019.0350,
  arXiv:1804.01734, 2020.

\bibitem{Palmer:2022b}
T.N. Palmer.
\newblock {\em The Primacy of Doubt}.
\newblock Oxford University Press, 2022.

\bibitem{Palmer:2022}
T.N. Palmer.
\newblock Quantum physics from number theory.
\newblock arXiv:2209.05549, 2022.

\bibitem{Pearl}
J.~Pearl.
\newblock Causal and counterfactual inference.
\newblock {\em The Handbook of Rationality. The MIT Press}, pages 427--438,
  2021.

\bibitem{Penrose:2014}
R.~Penrose.
\newblock On the gravitization of quantum mechanics 1: Quantum state reduction.
\newblock {\em Foundations of Physics}, 44:557--575, 2014.

\bibitem{PriceWharton}
H.~Price and K.~Wharton.
\newblock Entanglement swapping and action at a distance.
\newblock {\em Foundations of Physics}, 51:105, 2021.

\bibitem{Redhead}
R.~Readhead.
\newblock {\em Incompleteness Nonlocality and Realism}.
\newblock Oxford University Press, 1992.

\bibitem{Robert}
A.~M. Robert.
\newblock {\em A Course in p-adic Analysis}.
\newblock Springer ISBN 0-387-98660-3, 2000.

\bibitem{RollsDeco}
E.T.. Rolls and G.~Deco.
\newblock {\em The Noisy Brain}.
\newblock Oxford University Press, 2012.

\bibitem{Scarani:2019}
S.~Scanrani.
\newblock {\em Bell Nonlocality}.
\newblock Oxford Graduate Texts, 2019.

\bibitem{Schwartz:2019}
M.~Schwartz.
\newblock Statistical mechanics lecture 3.
\newblock https://scholar.harvard.edu/files/schwartz/files, 2019.

\bibitem{Susskind:2012}
L.~Susskind.
\newblock Fractal flows and time's arrow.
\newblock {\em arXiv:1203.6440}, 2012.

\bibitem{tHooft:2015b}
G.~'tHooft.
\newblock {\em The Cellular Automaton Interpretation of Quantum Mechanics}.
\newblock Springer, 2016.

\bibitem{Wheeler}
J.~A. Wheeler.
\newblock {\em Information, Physics, Quantum: The Search for Links}.
\newblock The Santa Fe Institute Press, 2023.

\bibitem{WisemanCavalcanti}
H.M. Wiseman and E.G.Cavalcanti.
\newblock Causarum investigatio and the two {B}ell's theorems of {J}ohn {B}ell.
\newblock arXiv:1503.06413, 2015.

\bibitem{Wolfram}
S.~Wolfram.
\newblock {\em A New Kind of Science}.
\newblock Wolfram Media, 2002.

\bibitem{Zeilinger}
A.~Zeilinger.
\newblock Zeilinger on superdeterminism.
\newblock
  https://www.physicsforums.com/threads/zeilinger-on-superdeterminism.742415,
  2014.

\bibitem{Spekkens}
E.~Wolfe Zjawin, B. and R.W. Spekkens.
\newblock Restricted hidden cardinality constraints in causal models.
\newblock arXiv:2109.05656, 2021.

\end{thebibliography}
\end{document}